\documentclass{article}

\usepackage{amsmath}
\usepackage{enumitem}
\usepackage{amsthm}

\newtheorem{theorem}{Theorem}
\newtheorem{lemma}{Lemma}

\newtheorem{definition}{Definition}
\newtheorem{conjecture}{Conjecture}

\usepackage{amsfonts}
\usepackage{amssymb}
\usepackage{graphicx}
\usepackage{appendix}
\usepackage{algorithm}
\usepackage{caption}

\usepackage[noend]{algpseudocode}

\usepackage{xcolor}

\title{Efficient Enumeration of Maximal Cliques in Weakly Closed Graphs}

\author{George Manoussakis \textsuperscript{1},\\
\textsuperscript{1} LI-PARAD, Universit\'e de Versailles Saint-Quentin-En-Yvelines, Paris Saclay}
\date{}

\begin{document}

\maketitle

\begin{abstract}

We demonstrate that the algorithm presented in [J. Fox, T. Roughgarden, C. Seshadhri, F. Wei, and N. Wein. Finding cliques in social networks:  A new distribution-free model. SIAM journal on computing, 49(2):448–464, 2020.] can be modified to achieve an enumeration time complexity of $\alpha\mathcal{O} (poly(c)n)$. Here, parameter $c$ represents the weakly closure of the graph, $\alpha$ denotes the number of maximal cliques, and $n$ refers to the graph's order. This result improves over their complexity, which is exponential in the closure of the graph. 

\end{abstract}

\section{Introduction}

Maximal clique enumeration finds many applications in various domains.  In social networks, finding  densely connected clusters of users can be used to predict how a virus spreads. Another application is targeted advertising: entities in a cluster have higher probability of responding positively to similar adds. The problem is, by definition, difficult since the number of maximal cliques can be exponential. Moreover, social networks graphs often have a significant number of nodes and edges. Thus, it seems relevant to find suitable parameters of the graph to achieve tractable complexity. There has been a long line of work following this approach and such results can be found in \cite{30,27,28,29,24,26,G1,25}, for instance. The parameters used in these papers often aim to capture the sparseness of the input graph, since many real-word graphs have few edges. This is especially true for social networks graphs. Well-known parameters include the maximum degree, arboricity or degeneracy~\cite{4}. A sparse graph has generally bounded degeneracy. In this context, Fox and al.~\cite{fox2020} introduced a parameter that captures the triadic closure principle in social networks. According to this principle, if two users have a significant number of mutual friends, it is highly likely that they are friends themselves. They call this parameter the closure number of the graph and they formally define it as follows:

\begin{definition}
For a positive integer $c$, an indirected graph $G=(V,E)$ is $c$-closed if, whenever two distinct vertices $u,v \in V$ have at least $c$ common neighbors, $(u,v)$ is an edge of $G$.
\end{definition}

As this definition is very restrictive, since a single pair of vertices may prevent the graph from being $c$-closed, it has been relaxed as follows :

\begin{definition}

Given a graph and a value of $c$, a bad pair is a non-adjacent pair of vertices with at least $c$ common neighbors. A graph is weakly $c$-closed if there exists an ordering of the vertices $\{v_1, v_2, \dots , v_n\}$ such that for all $i$, $v_i$
is in no bad pairs in the graph induced by $\{v_i, v_{i+1}, \dots , v_n\}$. 

\end{definition}

From a practical point of view, the weakly closure seems to be an interesting parameter as many real words social graphs have a small weakly closure~\cite{fox2020}. It tries and capture the triadic property principle discussed before which empirically usually holds~\cite{blo}. Since the closure and the weakly closure are computable in polynomial time~\cite{fox2020}, it would be interesting to develop fixed-tractable algorithms with respect to the closure of the input graph.

\subsection{Our contributions}

Fox and al.~\cite{fox2020} demonstrated that a $c$-closed graph has $\mathcal{O}(3^{c/3}n^2)$ maximal cliques. Building upon their proof, they devised an algorithm enumerating all the maximal cliques in these graphs in time $\mathcal{O}(3^{c/3}2^cn^2)$. This algorithm includes an initial preprocessing phase, which computes induced 2-paths known as wedges, and runs in polynomial time. Furthermore they established that these results naturally extend to weakly $c$-closed graphs. These findings have also been generalised to other clique-like structures by Koana and al.~\cite{koana}.

We will adopt the approach of~\cite{fox2020}
and present enhancements to their algorithm that lead to improved time complexities. Our first main result is the development of an algorithm that achieves a time complexity of $\mathcal{O}(poly(c)3^{c/3}n^2)$. This improvement removes the exponential factor of $2^c$ from the original complexity, resulting in a more efficient algorithm. This is stated in Theorem~\ref{algexp}. Furthermore, leveraging the same approach, we demonstrate that it is possible to develop an output-sensitive algorithm. Its complexity depends on the size of the solution. Specifically, our algorithm achieves an enumeration time complexity of $\alpha \mathcal{O}(npoly(c))$, with an additional polynomial time preprocessing phase. This is stated in Theorem~\ref{algosens}.

\begin{figure}[h]
   
\begin{algorithm}[H]
    \begin{algorithmic}[1]

\Procedure{\text{Preprocess}}{$G$}    
  \State Enumerate all wedges in $G$
   \State Let $M$ be a mapping such that $\forall x \in V(G)$, $M[x]$ is the set
    of wedges with $x$ as an endpoint
  
 \EndProcedure
\Procedure{\text{CClosedClique}}{$G$} 
  \If{$|V(G)|=1$}
        \State return $G$ \label{last}
    \EndIf
    \State Fix an arbitrary vertex $v\in V(G)$ \label{Vertex}
     \State Initialise $F$ an empty forest
    \State F $\leftarrow$ CClosedCliques($G \backslash \{v\}$) \label{recs}
    
    \State Given a leaf $l$ in $F$, let $K(l)$ be the set of vertices from the root to $l$
    \For {each leaf $l$ in $F$ with $K(l)\subseteq N(v)$}\label{dfs} 
    \State add $v$ to $F$ as a child of $l$ \label{child}
    \EndFor
    \State $E'\leftarrow$ the set of all edges in some wedge in $M[v]$ except for those edges incident to $v$
    \State $H \leftarrow G(V,E')$
    \State Initialise $F_u$ an empty forest
    \For{each vertex $u\in (G\backslash (N(v)\cup \{v\}))$}{\label{for}}
    
    \State $F_u \leftarrow CLIQUES(G[N_{H}(u)])$\label{tree1} \label{Enu}
    \State add $v$ to $F_u$ with an edge between $v$ and every current root in $F_u$\label{tree2}
    \EndFor
   
\State \Return {$F\cup (\cup_{u} F_u)$}

\EndProcedure

\end{algorithmic} 
\caption{:}
\label{alg1}
  \end{algorithm}
    
\caption{The algorithm of Fox and al.\cite{fox2020}. In this version, only a superset of the maximal cliques is outputted. Here wedges denote the induced paths of length two of the graph.}
\end{figure}

\section{Results}
\label{res}

The graphs we consider are of the form $G=(V,E)$ with $V$ the vertex set and $E$ the edge set. If $X\subset V$, the subgraph
of $G$ induced by $X$ is denoted by $G[X]$.  When not clear from the context, the vertex set of $G$ will be denoted by $V(G)$. The set $N(x)$ is called the \textit{open neighborhood} of the vertex $x$ and consists of the vertex adjacent to $x$ in $G$. \textit{The closed neighborhood} of $x$, denoted by $N[x]$,  is defined as the set $ \{N(x) \cup{x}\}$. Given an ordering $v_{1},...,v_{n}$ of the vertices of $G$, set $V_{i}$ is the vertices following $v_{i}$ including itself in this ordering, that is, the set $\{v_{i},v_{i+1},...,v_{n}\}$.
 Assume that vertices are considered recursively at Line~\ref{Vertex}  of Algorithm~\ref{alg1} in some ordering  $\sigma= v_1,v_2,\dots,v_i,\dots, v_n $. We define set $A_i=(V_i \cap (G\backslash (N(v_i)\cup \{v_i\})))$. That is, set $A_i$ is the set of vertices coming after $v_i$ in the ordering $v_1,v_2,\dots,v_n$ and which are not in its closed neighborhood. This set is crucial when computing the maximal cliques of the graph.
 
 \subsection{Description of the Algorithm of Fox and al.~\cite{fox2020}.}
In the original paper~\cite{fox2020}, two versions of the algorithm are presented.
 The first version, referred to as Algorithm~\ref{alg1}, provides only a superset of the solution. We give a complete description, for the sake of completeness. The algorithm is recursive and based essentially on the proof of the following theorem, as stated in the paper by Fox and al~\cite{fox2020}. Let $F(n,c)$ denote the maximum
possible number of maximal cliques in a $c$-closed graph on $n$ vertices.

\begin{theorem}
\label{fox}
For all positive integers $c$, $n$, we have $F(n, c) \leq 3^{c-1/3}n^2$.

\end{theorem}

The algorithm is derived from the constructive proof of this theorem which relies on the following step. Let $G$ be a $c$-closed graph on $n$ vertices and let $v \in V (G)$ be an arbitrary vertex. Every maximal clique $K \subseteq G$ is of one of the following types:
\begin{enumerate}
    \item The clique $K$ does not contain vertex $v$; and $K$ is maximal in $G \backslash \{v\}$.
    \item  The clique K contains vertex $v$; and $K \backslash \{v\}$ is maximal in $G \backslash \{v\}$.
    \item The clique $K$ contains vertex $v$; and $K \backslash \{v\}$ is not maximal in $G \backslash \{v\}$.
\end{enumerate}

\begin{figure}[h]
\includegraphics[width=0.8\textwidth]{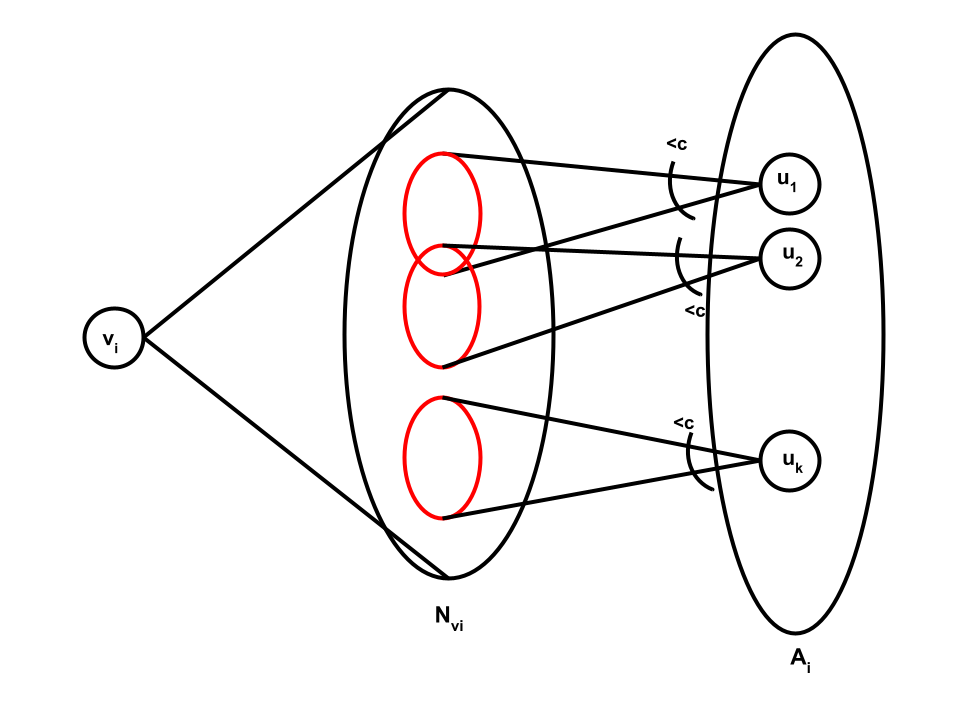}
\caption{A schematic description of the sets considered in a recursive step. Vertex $v_i$ is the current vertex in the recursion. Set $N_{v_i}$ represents its neighbors with higher rank in the closure ordering and $A_i$ is the set of its non neighbors with higher rank. The new cliques that must be computed at each level of the recursion are included in the red sets. Due to the $c$-closure property, they have size at most $c$.}
\label{sets}
\end{figure}

The approach described naturally lends itself to a recursive algorithm. The algorithm considers vertices in a recursive manner, constructing cliques starting from the last recursive call. In each recursive step, when a new vertex $c$ is considered, there are two types of cliques that need to be considered. The ones that have been already computed in older recursive calls (point 2) and the ones which were not maximal before vertex $c$ was added to the graph (point 3). The main contribution of Theorem~\ref{fox} is an algorithm computing cliques corresponding to point 3. Below is a description of the algorithm. It is worth noting that this decomposition for enumerating maximal cliques can be extended to other graph parameters, as discussed in the conclusion.

In Algorithm~\ref{alg1}, when considering a recursive call to some vertex $v$ in Line~\ref{Vertex}, the cases 2 and 3 are considered. Case 2 is examined in the \textsc{FOR} loop in Line~\ref{recs}. The algorithm tries to add the current vertex in the solution forest, to existing candidate cliques, which has been computed during precedent recursive calls. The difficulty is in Case 3, when new candidate cliques must be computed, since they have not been found so far in the recursive calls. The authors showed that the maximal cliques in the graph at the current recursive step can be found in the intersection of the neighborhood of the considered vertex (referred to as Line~\ref{Vertex}) and the neighborhood of its non-neighbors. A description is given in Figure~\ref{sets}. By the $c$-closure property of the graph, this intersection is of order smaller than $c$. This is done in the \textsc{FOR} loop of Line~\ref{for}. When such a clique is found, it is added to the solution forest $F_u$, as a new tree.  This procedure does not output exactly all the maximal cliques of the input graph but rather a superset of the solution. Some cliques that are outputted may not be maximal, but all the maximal cliques are included. To obtain the exact solutions, an additional step is required, as described by the authors below. All the cliques computed in Line~\ref{Enu} are added to the hash set $S$. Then, each clique $K$ in this set is considered, starting from largest to the smallest size. If any subset of $K$ is in $S$ the clique is rejected. Therefore, this additional step requires $\mathcal{O}(2^cc)$ time since cliques computed in Line~\ref{Enu} have a maximal size of $c$. The maximal cliques in the subgraphs are computed (in Line~\ref{Enu}) using a result from Tomita et al.~\cite{8}, which finds maximal cliques in a graph of order $n$ in time $\mathcal{O}(3^{n/3})$.

Taking all the factors into account, the complexity analysis for the algorithm that computes a superset of the maximal cliques of the graph, referred to as Algorithm~\ref{alg1}, is as follows. Assuming $T'(n,c)$ the running time of the procedure \textsc{CClosedClique}, then overall the algorithm requires $T'(n-1,c)$ time for the recursive call and an additional $\mathcal{O} (3^{c/3})$ for each CLIQUE call within the recursive step, using the algorithm of Tomita and al.~\cite{8}. Since there are at most $n$ calls to the algorithm of Tomita and al., this brings the total complexity to: $$T'(n,c)=T'(n-1,c)+\mathcal{O}(3^{c/3}n)=\mathcal{O}(3^{c/3}n^2)$$

As discussed earlier, an additional curing step is necessary to get exactly the set of maximal cliques of the graph. The rationale behind this step is explained in the following section.  Each clique requires a time complexity of $\mathcal{O}(2^cc)$ for the curing process. Hence, the overall time complexity to get exactly all the maximal cliques is $\mathcal{O}(3^{c/3}2^{c}cn^2)$. Furthermore, an additional preprocessing step is involved, which entails listing all induced 2-paths in a $c$-closed graph of order $n$. This step takes $\mathcal{O}(p(n,c))$ time, where $p(n,c)$ represents the time required for the listing process. Instead of explicitly outputting a solution, the algorithm stores the cliques in a tree structure constructed during the recursion. Similarly, our solutions in Theorem~\ref{algosens} and Theorem~\ref{algexp} will follow a similar format.

\subsection{Description of our Algorithm.}

We propose an algorithm which is inspired by Algorithm~\ref{alg1}.  It enumerates maximal cliques in each recursive step using a different approach than the original method by Fox et al. Our solution introduces a new algorithm called Algorithm~\ref{alg2}, the code for which is provided below. The FOR loop of Line~\ref{Enu} of Algorithm~\ref{alg1} has been replaced with our main contribution which is the \textsc{DoubleScan} procedure at Line~\ref{dbls} of Algorithm~\ref{alg2}.
We consider graphs $G_{v_i}$ for $i\in [n]$, formed of the union of the graphs $G[N_H(u)]$ with $u$ belonging to the set $A_i$. Here $v_i$ is the vertex currently considered in the recursive call (at Line~\ref{Vertex}). We then want to compute exactly the maximal cliques of graph $G_{v_i}$ for each $i\in [n]$. The problem that appeared in the original paper~\cite{fox2020} when enumerating the cliques was the following.  There could be a clique $K$ in $N(v_i) \cap N(u_j) \cap N(u_k)$ for some $u_j$ and $u_k$ in $A_i$
such that $K$ is maximal in $N(v_i) \cap N(u_k)$ but not maximal in $N(v_i) \cap N(u_j)$. In this case $K \cup {v_i}$
will be reported as a maximal clique of $G_{v_i}$ even though it is not  maximal. We state this in Theorem~\ref{mainth} which is part of the proof of the algorithm in the paper of Fox and al.~\cite{fox2020}.

\begin{theorem}[Fox and al.~\cite{fox2020}]
   \label{mainth} 
   Every  clique  computed in Line~\ref{for} of Algorithm~\ref{alg1} belongs to a graph $G[N_H(u)]$.  If there is such a clique $K$ which is maximal in graph $G[N_H(u)]$  but not in $G_v$ then there exists $v\neq u$ such that $K \subseteq V(G[N_H(v)])$.
    
\end{theorem}

There are two main differences in our approach compared to the original paper by Fox et al.~\cite{fox2020}. Firstly, to achieve output-sensitive complexities, we deviate from using the result presented by Tomita et al.~\cite{8}. Instead, we employ output-sensitive algorithms specifically designed for maximal clique enumeration, such as those proposed in references \cite{30,27,28,29,24,26,G1,25}. Secondly, we will enhance the curing step by introducing modifications to the data structures and incorporating combinatorial arguments. These adjustments will enable us to accurately enumerate all maximal cliques, resulting in an improved overall complexity and validating the claimed results.
We consider the vertices  $V(G_{v_i})$ ordered as follows. Assume that when constructing graph $G_{v_i}$, vertices $u\in A_i$ are considered in some ordering $\sigma_{v_i}=u_1,u_2,\dots,u_k$. Then we construct an ordering $\sigma_{c}$ as follows. We put the vertices  $V(G[N_H(u_1)])$ in any order followed by the vertices  $V(G[N_H(u_2)])\backslash V(G[N_H(u_1)])$ in any order followed by the vertices  $V(G[N_H(u_3)])\backslash (V(G[N_H(u_1)])\cup V(G[N_H(u_2)]))$ in any order and so on for each vertex $u_i, i\in [k]$. We will also consider at some point in the algorithm the reverse ordering that we call $\overline{\sigma_{c}}$. Once this is done, we compute the maximal cliques of graph $G_{v_i}$ as described in the \mbox{\textsc{DoubleScan}} procedure.

  \begin{figure}[h]
\begin{algorithm}[H]
  
\floatname{algorithm}{\textsc{Algorithm}}
    \begin{algorithmic}[1]

\Procedure{\text{Preprocess}}{$G$}    
  \State Enumerate all wedges in $G$
    \State Let $M$ be a mapping such that $\forall x \in V(G)$, $M[x]$ is the set
    of wedges with $x$ as an endpoint
  
 \EndProcedure
\Procedure{\text{CClosedClique2}}{$G$} 
  \If{$|V(G)|=1$}
        \State return $G$ 
    \EndIf
    \State Fix an arbitrary vertex $v\in V(G)$ 
    \State F $\leftarrow$ \textsc{CClosedClique2}($G \backslash \{v\}$) 
    \For {each leaf $l$ in $F$ with $K(l)\subseteq N(v)$}
    \State add $v$ to $F$ as a child of $l$ \label{test}
    \EndFor
   \State $E'\leftarrow$ the set of all edges in some wedge in $M[v]$ except for those edges incident to $v$
    \State $H \leftarrow G(V,E')$
    \State $R \leftarrow$ \textsc{DoubleScan}$(H)$\label{dbls}
\State \Return $F\cup R$

\EndProcedure

\end{algorithmic} 
  \caption{:}
 \label{alg2}
\end{algorithm}
\caption{This algorithm is a modification of Algorithm 1 of Fox and al.~\cite{fox2020}. Here the cliques are enumerated in each recursive step using the \textsc{DoubleScan} procedure, which is the main contribution of the paper.}

\end{figure}

\begin{figure}[h]

\begin{algorithm}[H]

\renewcommand{\thealgorithm}{}
\floatname{algorithm}{\textsc{DoubleScan(H)}}

    \begin{algorithmic}[1]

\State Initialise two prefix trees $T_{start},T_{end}$.
\State Initialise a set of trees $T_{check}$
\State Let $\sigma_{v_i}=u_1,u_2,\dots,u_k$.

    \For{($i=1$, $i \leq k$, $i++$)}

  \For{each clique $K$ in $G[N_H(u_i)]$}
  \State order the vertices of $K$ following $\sigma_c$
  \If{$K$ is not in $T_{start}$}
  \State Add it to $T_{start}$
  \EndIf
   \For {any prefix $P_K$ of $K$}
  \If {$P_K$ is a leaf in $T_{start}$ }
  \State remove it from $T_{start}$ \label{delstart}
 \EndIf
  \EndFor
   \EndFor
    
    \EndFor
    
      \For{($i=k$, $i \geq 1$, $i--)$}
  \For{each clique $K$ in $G[N_H(u_i)]$}
  \State order the vertices of $K$ following $\overline{\sigma_c}$
  \If{$K$ is not in $T_{end}$}
  \State Add it to $T_{end}$
  \EndIf
   \For {any prefix $P_K$ of $K$}
  \If {$P_K$ is a leaf in $T_{end}$ }\label{rev}
  \State order its vertices following ${\sigma_c}$
  \State add it to $T_{check}$
 \EndIf
  \EndFor
   \EndFor
    
    \EndFor

  \For{each element $K$ of $T_{check}$} 
  
 \If{$K$ is in $T_{start}$}
 \State remove it from $T_{start}$
 \EndIf
  \EndFor
    
\State add every clique corresponding to a leaf of $T_{start}$ to set $F_H$ as a tree
    
\State add $v$ to $F_H$ with an edge between $v$ and every current root in $F_H$
    
\State \Return {$ F_H$}

\end{algorithmic}

 \caption{:}
  \label{proc}
\end{algorithm}

\caption{This procedure enumerates the maximal cliques in each recursive step of the algorithm. Moreover, it checks for incorrect maximal cliques by storing candidate problematic cliques in a suffix tree.}
\end{figure}

\subsection{Proof of Correctness and time complexity.}

Note first that in the original algorithm, the FOR loop at Line~\ref{for} iteratively enumerates all maximal cliques of the graphs $G[N_H(u_i)]$. However, as discussed earlier, the algorithm computes a superset of these cliques, and additional steps are required to obtain precisely the maximal cliques we desire. Our objective is to demonstrate that the \textproc{DoubleScan} procedure precisely computes the set of maximal cliques for the graphs $G[N_H(u_i)]$. The steps of the proof are as follows. If ones enumerates the cliques of the graph $G[N_H(u_i)]$ iteratively, the problem that appears is that a maximal clique can be maximal in some graph $G[N_H(u_i)]$ but not maximal in the graph if its vertex set is included in the vertex set of some other graph $G[N_H(u_j)]$ with $j \neq i$. The \textsc{DoubleScan} procedure specifically checks the outputted maximal cliques which are in the intersections of these subgraphs. It does this in two steps. First it checks all the cliques in the subgraphs iteratively starting from the first one, considering the order $\sigma_c$. Then it essentially does the same thing but in the reverse ordering (in the second for loop). This is necessary since some maximal cliques in the subgraphs which are not maximal in the graph may have been missed during the first scan.

\begin{theorem}
\label{cor}
    The \textsc{DoubleScan} procedure outputs exactly all maximal cliques of graph $G_{v_i}$.
\end{theorem}

\begin{proof}
Assume by contradiction that the procedure outputs some clique $K$ which is not maximal in graph $G_{v_i}$, where we recall that $v_i$ is the vertex of the recursive step at Line~\ref{Vertex}. By construction, this clique belongs to a graph $G[N_H(u_i)]$ for some $i\in [k]$. Using Theorem~\ref{mainth}, there are two cases. First there exists a graph $G[N_H(u_j)]$ with $j> i$ (the indexes in $\sigma_{v_i}$) such that $K$ included in $V(G[N_H(u_j)])$ and $K$ not maximal in $G[N_H(u_j)]$. Or in the second case there exists a graph $G[N_H(u_j)]$ with $j< i$ such that $K$ included in $V(G[N_H(u_j)])$ and $K$ not maximal in $G[N_H(u_j)]$. Consider the first case. Assume that the vertices of $K$ are ordered following $\sigma_c$. If $K$ is maximal in $G[N_H(u_i)]$ but not in $G[N_H(u_j)]$ with $j\geq i$ then there exists a set $S\subseteq V(G[N_H(u_j)])$ of vertices such that $K\cup S$ is maximal in $G[N_H(u_j)]$. Note that if we order the vertices of $K\cup S$ following $\sigma_c$ then the vertices of $K$ are a prefix of clique $K\cup S$. Thus, if $K$ has been stored in the prefix tree $T_{start}$ it will be deleted at Line~\ref{delstart} of the procedure, as it is maximal in $G[N_H(u_i)]$ but not in $G[N_H(u_j)]$. This gives the proof for the first case. The second case is similar. Except that we scan graphs $G[N_H(u_i)]$ for $i$ from indices $k$ to $1$ and order the vertices of the cliques following $\overline{\sigma_{c}}$. This allows us to find cliques which are maximal in some graph $G[N_H(u_i)]$ but not in a graph $G[N_H(u_j)]$ with $j< i$. Once such a clique is found in Line~\ref{rev} it is stored in the set $T_{check}$. Then every clique in in this set is reordered following $\sigma_{c}$ and is removed from $T_{start}$ if it appears there. Finally once this is done, tree $T_{start}$ contains exactly every maximal clique of graph $G_{v_i}$.
\end{proof}

\begin{lemma}
\label{exp}
   Using the algorithm of Tomita et al.~\cite{8} for maximal clique enumeration, The \textproc{DoubleScan} procedure has time complexity $\mathcal{O}(poly(c)3^{c/3}n)$.
\end{lemma}

\begin{proof}

Computing all maximal cliques of graph $G_{v_i}$ requires time $\mathcal{O}(3^{c/3}n)$ since there are at most $n$ graphs of cardinality bounded by $c$, with $c$ the closure of the graph. After that, the vertices of each clique are ordered. Then, we search, insert or delete the cliques from a prefix tree. These operations cost at most $\mathcal{O}(c\log c)$ for each clique since these cliques are of order at most $c$.   There can be at most $\mathcal{O}(3^{c/3})$ maximal cliques in a graph of order $c$, by the Moon and Moser~\cite{MM} bound. Thus, in total the complexity of the procedure is  $\mathcal{O}( poly(c)3^{c/3}n)$.

\end{proof}

In the next lemma, we demonstrate how utilizing output sensitive algorithms for maximal clique enumeration in the subgraphs in each recursive step leads to a total time complexity that is output sensitive. To accomplish this we need to examine in detail how the complexities accumulate at each step. We will consider each graph $G_i$ with $i$ the current vertex in the recursive call. We will bound the number of maximal cliques of the graphs $G[N_H(u_i)]$ for $i\in [k]$ with respect to the number of maximal cliques of graph $G_i$.

\begin{lemma}
\label{lemsens}
    Using an output sensitive algorithm for maximal clique enumeration, The \textproc{DoubleScan} procedure has time complexity $\alpha_i\mathcal{O}(poly(c))+\mathcal{O}(poly(c)n)$ with $\alpha_i$ the number of maximal cliques of graph $G_i$.
\end{lemma}

\begin{proof}
Let $\beta_i$ be the number of cliques of graph $G[N_H(u_i)]$, for each $i\in [k]$. As shown in Theorem~\ref{mainth} and in the proof of Theorem~\ref{cor}, if a clique $C$ is maximal in some graph $G[N_H(u_i)]$ but not in $G_i$ then there exists a maximal clique $K$ in some graph $G[N_H(u_j)]$ with $j\neq i$ such that the ordered vertices of $C$ form a prefix of the vertices of the ordered vertices of $K$ (following $\sigma_c$ or $\overline{\sigma_c}$). Therefore the number of maximal cliques of graphs $G[N_H(u_i)]$, $i\in [k]$, which is $\sum_{i=1}^k \beta_i$, is bounded by $nc\alpha_i$, with $\alpha_i$ the number of maximal cliques of $G_i$. This is due to the following reasons. First observe that there can be at most $c$ different prefixes of a given maximal clique. This gives $c\alpha_i$ different prefixes in total. Then each such suffix may appear in each graph $G[N_H(u_i)]$ for~$i\in [k]$, meaning it can be computed up to $k$ times. Since $k$ is bounded by $n$, we get the upper bound.

We enumerate all maximal cliques of each graph $G[N_H(u_i)]$ using an output sensitive algorithm, see \cite{30,27,28,29,24,26,G1,25} for instance. Usually, there is a preprocessing phase and an enumeration phase. Since the graphs are order $c$, the preprocessing phase takes $\mathcal{O}(poly(c))$. Similarly the enumeration phase is done in $\beta_i\mathcal{O}(poly(c))$. In total the enumeration for all graphs $G[N_H(u_i)]$, $i\in [k]$ is done in time $\sum_{i=1}^k\beta_i\mathcal{O}(poly(c))+\mathcal{O}(poly(c)n)$ since $k$ is bounded by $n$. Then, the procedure deals with each clique in time $\mathcal{O}(poly(c))$, similarly to the proof of Lemma~\ref{exp}. In conclusion, since $\sum_{i=1}^k \beta_i$ is bounded by $nc\alpha_i$, the total complexity is $\alpha_i\mathcal{O}(poly(c))+\mathcal{O}(poly(c)n)$.
\end{proof}

\begin{theorem}
\label{algexp}
Using the algorithm of Tomita et al.~\cite{8} for maximal clique enumeration, Algorithm~\ref{alg2} has enumeration time complexity $\mathcal{O}(poly(c)3^{c/3}n^2)$. It has an additional $\mathcal{O}(poly(c)n^2+p(n,c))$ preprocessing time complexity.
\end{theorem}

\begin{proof}
The proof is essentially the same as in the original paper of Fox and al.~\cite{fox2020}. Assuming $T'(n,c)$ the running time of the procedure \textsc{CClosedClique2}, then overall the algorithm requires $T'(n-1,c)$ time for the recursive call and an additional $\mathcal{O}(poly(c)3^{c/3}n)$ time for the call to the \textsc{DoubleScan} procedure, as shown in Theorem~\ref{exp}. This brings the total time complexity to: $$T'(n,c)=T'(n-1,c)+\mathcal{O}(poly(c)3^{c/3}n)+\mathcal{O}(poly(c)n)=\mathcal{O}(poly(c)3^{c/3}n^2)+\mathcal{O}(poly(c)n^2)$$

The additional $\mathcal{O}(p(n,c))$ comes from the computation of the induced $2$-paths.
\end{proof}

\begin{theorem}
\label{algosens}
  Using an output sensitive algorithm for maximal clique enumeration, algorithm~\ref{alg2} has enumeration time complexity $\alpha_i \mathcal{O}(poly(c))$ with $\alpha_i$ the number of maximal cliques of graph $G_i$. It has an additional $\mathcal{O}(poly(c)n^2+p(n,c))$ preprocessing time complexity.
\end{theorem}

\begin{proof}
Using Lemma~\ref{lemsens}, the recursive formula is as follows in that case :

$$T'(n,c)=T'(n-1,c)+\alpha_n \mathcal{O}(npoly(c))+\mathcal{O}(poly(c)n)$$

where $\alpha_n$ is the number of maximal cliques of graph $G_{v_n}$. In total this is equal to :

$$T'(n,c)=(\sum_{i=1}^{n}\alpha_i \mathcal{O}(npoly(c)))+\mathcal{O}(poly(c)n^2)$$
In the final solution tree, a path from the root to a leaf is a maximal clique of a graph. Observe that the only time a path is added to the solution tree during the execution of Algorithm~\ref{alg2} is during the execution of the \textsc{DoubleScan} procedure. Therefore the number of maximal cliques of the input graph is equal to the sum of the number of maximal cliques of graphs $G_{v_i}$, with $i\in [n]$. This gives the following complexity, which concludes the proof :

$$T'(n,c)=\alpha \mathcal{O}(npoly(c))+\mathcal{O}(poly(c)n^2)$$

The additional $\mathcal{O}(p(n,c))$ comes from the computation of the induced $2$-paths.

\end{proof}

\subsection{Application to other graph parameters}
\label{conc}

The decomposition technique used to prove Theorem~\ref{fox}, which we previously discussed is  interesting on its own and can be applied to other graph parameters, for maximal clique enumeration. For example, we can naturally utilize the degeneracy of the graph. Recall that a graph has degeneracy $k$ if there is a vertex of degree at most $k$ in any induced subgraph. This gives rise to a degeneracy ordering, an ordering of the vertices of the graph, where any vertex has at most $k$ neighbors with higher degeneracy.
 This property can be applied within the context of the decomposition given by Fox et al. to enumerate all the maximal cliques in a graph with degeneracy $k$. We state the following conjecture.

\begin{conjecture}
    Algorithm~\ref{alg2} can be modified to enumerate all maximal cliques in a $k$-degenerate graph in time $\alpha \mathcal{O}(poly(k))$.
\end{conjecture}
The proposed algorithm would take advantage of the bounded number of neighbors with higher rank and incorporate modifications to the \textsc{DoubleScan} procedure, specifically in the way maximal cliques are enumerated. This approach aims to fully take advantage of the bounded size of the search space. Based on these observations, it is highly likely that the conjecture holds. It is worth noting that this complexity matches the bound presented by Manoussakis~\cite{G1}. In a more general sense, similar results are possible for any parameter that enables an ordering of vertices, allowing for the local enumeration of maximal cliques within the neighborhood of vertices with higher rank.

\bibliographystyle{plain}
\bibliography{bibli}

\end{document}